\documentclass[journal]{IEEEtran}
\usepackage{cite}
\usepackage{amsmath}
\usepackage{algorithmic}
\usepackage{array}

\usepackage{algorithm}
\usepackage{amsthm, amssymb, bbm}
\usepackage{graphicx}
\usepackage[shortlabels]{enumitem}
\usepackage{subfigure}
\usepackage{xcolor}

\newtheorem{definition}{Definition}
\newtheorem{theorem}{Theorem}

\newtheorem{lemma}[theorem]{Lemma}

\begin{document}

\title{Matching-Theory-Based Multi-User Cooperative Computing Framework}
\author{Ya~Zhou, Guopeng~Zhang, Kezhi~Wang and Kun~Yang
\thanks{This work was supported in part by the National Natural Science Foundation of China under Grant 61971421, Grant 61620106011, Grant U1705263, Grant 61871076; and the Fundamental Research Funds for the Central Universities (Grant ZYGX2019J001). \textit{(Corresponding author: Guopeng Zhang.)}}
\thanks{Ya Zhou and Guopeng Zhang are with the Engineering Research
Center of Mine Digitalization, Ministry of Education, and the School of Computer Science and Technology, China University of Mining and Technology, Xuzhou 221116, China (e-mail: zhou\underline{ }ya@cumt.edu.cn; gpzhang@cumt.edu.cn).}
\thanks{Kezhi Wang is with the Department of Computer and Information Science, Northumbria University, Newcastle NE1 8ST, U.K. (e-mail: kezhi.wang@northumbria.ac.uk).}
\thanks{Kun Yang is with School of Electronic and Information Engineering, Nanjing University of Information Science and Technology, Nanjing, China, and also with School of Information and Communication Engineering, University of Electronic Science and Technology of China, Chengdu, China. (e-mail: kunyang@uestc.edu.cn).} }

\maketitle

\begin{abstract}
In this paper, we propose a matching theory based multi-user cooperative computing (MUCC) scheme to minimize the overall energy consumption of a group of user equipments (UEs), where the UEs can be classified into the following roles: resource demander (RD), resource provider (RP), and standalone UE (SU). We first determine the role of each UE by leveraging the roommate matching method. Then, we propose the college admission based algorithm to divide the UEs into multiple cooperation groups, each consisting of one RP and multiple RDs. Next, we propose the rotation swap operation to further improve the performance without deteriorating the system stability. Finally, we present an effective task offloading algorithm to minimize the energy consumption of all the cooperation groups. The simulation results verify the effectiveness of the proposed scheme.
\end{abstract}


\section{Introduction}
\IEEEPARstart{W}{ith} the rapid deployment of the computationally-intensive tasks, e.g., virtual reality, the requirement for UE in terms of battery life and computing resource are also increasing. Although traditional fixed infrastructure-based mobile edge computing (MEC) may help to provide computing resources to UEs, they may be inaccessible in some situations, like disasters or emergency cases where infrastructures are unavailable. MUCC \cite{chen_2017_ICC} has recently been proposed to allow a UE to utilize available computing resources from neighboring UEs. In the framework of MUCC, the UEs can be classified as one of the following roles: 1) RD, which has a computationally-intensive task to be processed; 2) RP, which has available computing resource to provide for other RDs; and 3) SU, which may do the task itself.

A one-to-one MUCC scheme was proposed in \cite{lin_2019_coml} to minimize the long-term energy consumption of two UEs. In \cite{wu_2018_GC}, a RD has the option to offload the task to a MEC server or a RP.
In \cite{cao_2019_iot}, a one-to-multiple MUCC scheme was proposed to allow an RD to partition its task into multiple parts and offload them to multiple RPs for parallel execution.
In \cite{sheng_2018_cc}, an RD is allowed to select one RP from a group of potential RPs. The purpose of the above works is mainly to reduce the energy consumption or execution delay of individual users. However, they do not address the following issues: (1) how to determine the role of a UE, i.e., as an RD, an RP, or a SU; (2) how to determine the association between the RDs and the RPs; and (3) how each RP allocates its computing and communication resources to the served RDs.

Against the above background, the main contribution of this paper is as follows: (1) The aim is to minimize the overall energy consumption of all UEs in a MUCC system through optimizing the user association and resource allocation. This problem is formulated as a mix integer nonlinear programming (MINLP) problem; (2) Due the high computational complexity of MINLP, we employ matching theory to design a low-complexity algorithm for finding the suboptimal solution of the problem. Specially, the energy consumption-related user preference functions are defined to solve the user role assignment and user association problems under the matching theory framework. Furthermore, a successive convex approximation (SCA) based task offloading algorithm is proposed to minimize the energy consumption of the whole system.

\section{System Model}
\begin{figure}[!t]
  \centering
  \includegraphics[width=0.85\linewidth, height=0.25\linewidth]{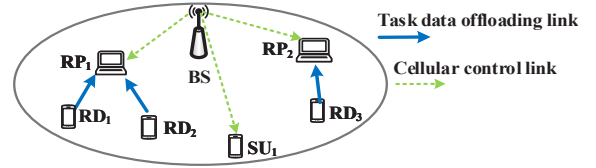}
  \caption{The task offloading process in a cooperation group.}
  \label{fig.system_model}
\end{figure}

As shown in Fig. \ref{fig.system_model}, we consider a MUCC system consisting of a set $\mathcal{N} \triangleq \left\{ 1, 2, \dots, N \right\} $ of $N$ UEs which can establish connections with each other via direct communications. Let $\mathcal{T} \triangleq \left\{ t = 1, 2, ... \right\}$ denote the sequence of the time slots. The length of slot $t$ ($ \forall t \in \mathcal{T} $) is $ \tau $. In slot $t$, UE $m$ ($ \forall m \in \mathcal{N} $) has a computing task to execute, which is represented by $\phi_m = \left( L_m, C_m \right)$, where $L_m$ (in bits) denotes the amount of task data to be processed, whereas $C_m$ (in CPU cycles/bit) denotes the number of CPU cycles required to be executed for each data bit. To complete task $\phi_m$ in a slot, the CPU frequency of UE $m$ is adjusted to $f_m=C_m L_m / \tau$, by using the dynamic voltage and frequency scaling (DVFS) technique \cite{lin_2019_coml}. Let $\gamma_m$ denote the effective capacitance coefficient of the CPU of UE $m$. Then, the energy spent by UE $m$ is as
\begin{equation}
  E^\mathbbm{S}_m = \gamma_m f_m^2 C_m L_m =  \gamma_m C_m^3 L_m^3 / \tau^2, \;\; \forall m \in \mathcal{N}.
  \label{equ.computing_energy}
\end{equation}

In any slot $t$, the UEs in set $\mathcal{N}$ can be divided into the following three subsets, i.e., the RD set $\mathcal{N}^{\mathbbm{D}}$, the RP set $\mathcal{N}^{\mathbbm{P}}$, and the SU set $\mathcal{N}^{\mathbbm{S}}$. Then, we have
\begin{align}
    & \mathcal{N}^x \cap \mathcal{N}^y = \emptyset, \;\; \forall x, y \in \left\{ \mathbbm{D}, \mathbbm{P}, \mathbbm{S} \right\}\ \text{and} \; x \neq y, \label{equ.cons_userset_1} \\
    & \mathcal{N}^{\mathbbm{D}} \cup \mathcal{N}^{\mathbbm{P}} \cup \mathcal{N}^{\mathbbm{S}} = \mathcal{N}. \label{equ.cons_userset_2}
\end{align}

Let $i$  $(\forall i \in \mathcal{N}^{\mathbbm{D}})$ denote the index of an RD, and $j$ $(\forall j \in \mathcal{N}^{\mathbbm{P}})$ the index of an RP. Any RD $i$ can partition its task $\phi_i$ into two parts. One part is processed locally, while the other part, with data size $l_{i,j}$ is offloaded to RP $j$ for cooperative computing. Then, we have
\begin{equation}
    0 \leq l_{i, j} \leq L_i, \;\; \forall i \in \mathcal{N}^{\mathbbm{D}}, \; \forall j \in \mathcal{N}^{\mathbbm{P}}.
    \label{equ.sr_offloading_size_cons}
\end{equation}

We define the indicator function $\mathbbm{1}_{0 \leq l_{i, j} \leq L_i} \in \left\{ 0, 1 \right\}$ to represent the association between the RDs and the RPs. If RD $i$ offloads part of its task $\phi_i$ to RP $j$, $\mathbbm{1}_{0 < l_{i, j} \leq L_i} = 1$, otherwise, $\mathbbm{1}_{l_{i, j} = 0} = 0$. In addition, RD $i$ can choose at most one RP $j$ to offload task $\phi_i$, that is
\begin{equation}
  \sum_{j \in \mathcal{N}^{\mathbbm{P}}} \mathbbm{1}_{l_{i, j}} \leq 1, \; \forall i \in \mathcal{N}^{\mathbbm{D}},
  \label{equ.sr_offloading_cons}
\end{equation}
In order to avoid allocating too many RDs to RP $j$ in a time slot, we preset a parameter $a_j$ $(a_j \geq 1)$ representing the maximum number of RDs that can be served by RP $j$ in a slot, that is
\begin{equation}
  \sum_{i \in \mathcal{N}^{\mathbbm{D}}} \mathbbm{1}_{l_{i, j}} \leq a_j, \; \forall j \in \mathcal{N}^{\mathbbm{P}}.
  \label{equ.sp_connection_cons}
\end{equation}
The set of RDs that are associated with RP $j$ can thus be denoted by $\mathcal{M}_j \triangleq \left\{ i \in \mathcal{N}^{\mathbbm{D}} \; | \; \mathbbm{1}_{l_{i, j}} = 1 \right\}$.

Since the UEs have limited battery energy, any UE $j$ can be used as a RP only when the available energy $E_j^{\text{Ax}}$ is greater than the threshold $E_j^{\text{min}}$. So we get the following constraint
\begin{equation}
    E_j^{\text{Ax}} \geq E_j^{\text{min}}, \; \forall j \in \mathcal{N}^{\mathbbm{P}}. \label{fig.total_cons}
\end{equation}

Next, we introduce the task offloading process in a cooperation group composed of RP $j$ and the associated RDs in set $\mathcal{M}_j$. As illustrated in Fig. \ref{fig.time_structure}, the task offloading will span the following two consecutive slots.

\begin{figure}[!t]
  \includegraphics[width=\linewidth]{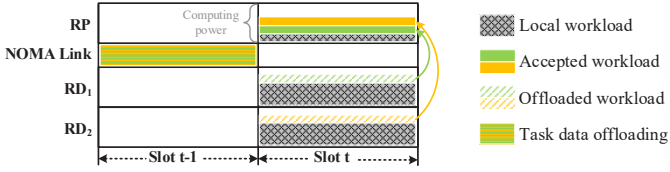}
  \caption{The task offloading process in a cooperation group.}
  \label{fig.time_structure}
\end{figure}

\underline{\textbf{Slot $t-1$}} (\emph{Task data offloading}):
In slot $t-1$, all the RDs in set $\mathcal{M}_j$ simultaneously transfer their task data to RP $j$ by using non-orthogonal multiple access (NOMA) \cite{wu_2018_GC}, while the system allocates orthogonal channels (of the same size of $w$ MHz) for different RPs to receive the task data. Let $p_{i,j}$ denote the transmit power of RD $i$ associated with RP $j$, and $g_{i,j}$ denote the channel gain between RD $i$ and RP $j$. When successive interference cancellation (SIC) \cite{wu_2018_GC} is used, the achievable data rate of RD $i$ is given by
\begin{align}
  r_{i, j} = w \log_2 \left( 1 + \frac{p_{i, j} g_{i, j}}{\sum_{n \in \mathcal{M}_j \backslash \left\{ i \right\}} p_{n, j}g_{n, j} + \sigma^2} \right),
\end{align}
where $\sigma^2$ is the noise power at the receiver of RP $j$. For correct receiving the task data of RD $i$ of $l_{i,j}$ bits, the lower bound of the transmit power of RD $i$ is obtained as \cite{wu_2018_GC}
\begin{align}
  p_{i, j} = \frac{\sigma^2}{g_{i, j}} \left( 2^{ \frac{l_{i, j}}{\tau w} } - 1 \right) 2^{\frac{1}{\tau w} \sum_{n \in \mathcal{M}_j \backslash \{ i \} } l_{n, j}} \leq P_i^\text{max}, \label{equ.transmit_power}
\end{align}
where $P_i^\text{max}$ is the maximum transmit power of RD $i$.

By using eq. \eqref{equ.transmit_power}, the energy consumption of RD $i$ in slot $t-1$ for task data offloading is obtained as
\begin{equation}
  E^\text{Tx}_{i, j} = p_{i, j} \tau, \;\; \forall i \in \mathcal{M}_j, \; \forall j \in \mathcal{N}^{\mathbbm{P}}.
  \label{equ.transmit_energy}
\end{equation}

\underline{\textbf{Slot $t$}} (\emph{Cooperative computing}): After collecting the task data of all the RDs in $\mathcal{M}_j$, RP $j$ can increase the CPU frequency as $f_j = \left( L_j + \sum_{i \in \mathcal{M}_j} \mathbbm{1}_{l_{i, j}} l_{i, j} \right) / \tau$ for completing the accepted tasks as well as its own task in slot $t$. The energy spent by RP $j$ for processing the tasks is given by
\begin{equation}
  E^\mathbbm{P}_j = \gamma_j C_j^3 \left( L_j + \sum_{i \in \mathcal{M}_j} \mathbbm{1}_{l_{i, j}} l_{i, j} \right)^3 / \tau^2, \;\; \forall j \in \mathcal{N}^\mathbbm{P}.
\end{equation}
Since the size of the result is small, the energy spent by RP $j$ for feeding back the result is ignored.

In the meanwhile, RD $i$ $(\forall i \in \mathcal{M}_j)$ can save its energy by decreasing the CPU frequency to $f_i = \left(L_i - \sum_{j \in \mathcal{N}_\mathbbm{P}} \mathbbm{1}_{l_{i, j}} l_{i, j} \right) / \tau$ for processing the remaining part of task $\phi_i$. Therefore, the energy spent by RD $i$ for local execution in slot $t$ is given by
\begin{equation}
  E^\text{Cx}_i = \gamma_i C_i^3 \left( L_i - \sum_{j \in \mathcal{N}_\mathbbm{P}} \mathbbm{1}_{l_{i, j}} l_{i, j} \right)^3 / \tau^2, \;\; \forall i \in \mathcal{N}^\mathbbm{D}.
\end{equation}

Then, the total energy spent by RD $i$ to complete task $\phi_i$ is given by
\begin{equation}
  E^\mathbbm{D}_i = E^\text{Cx}_i +\sum_{j \in \mathcal{N}^{\mathbbm{P}}} \mathbbm{1}_{l_{i, j}} E^\text{Tx}_{i, j}, \;\; \forall i \in \mathcal{N}^{\mathbbm{D}}.
\end{equation}

\section{Problem Presentation and Solution}
Our objective is to minimize the overall energy consumption of all the UEs in the system, which can be given by
\begin{alignat}{3}
    &\min_{l_{i, j}, \; \mathbbm{1}_{l_{i, j}}} \;\; && \sum_{i \in \mathcal{N}^{\mathbbm{D}}} E^\mathbbm{D}_i + \sum_{j \in \mathcal{N}^{\mathbbm{P}}} E^\mathbbm{P}_j + \sum_{n \in \mathcal{N}^{\mathbbm{S}}} E^\mathbbm{S}_n & \label{equ.min_function} \\
    &\mbox{s.t.} && \eqref{equ.cons_userset_1}, \; \eqref{equ.cons_userset_2}, \; \eqref{equ.sr_offloading_size_cons}, \; \eqref{equ.sr_offloading_cons}, \; \eqref{equ.sp_connection_cons}, \; \eqref{fig.total_cons}, \; \eqref{equ.transmit_power}. \notag
\end{alignat}

Problem \eqref{equ.min_function} is a MINLP, which is very difficult to solve in general. Next, we present an effective algorithm to find the sub-optimal solution of this problem based on matching theory. Note that we first do not consider constraint (7), thus allowing any UE to be used as a RP no matter how much energy it remains. This constraint will be dealt with in Sec. III-D by proposing a user role control operation.
Furthermore, it should be noted that the cooperative UEs in slot $(t-1)$ can perform the computing tasks accepted in slot $(t-2)$, and the communication link in slot $t$ can be used to transfer the task data required in slot $(t+1)$. Hence, problem (14) minimizes the energy consumption of the cooperative UEs in each slot.

\subsection{User role assignment}
From the view of one-to-one roommate  matching \cite{irving_1985_algorithm}, an RP and the associated RD can be seen as a pair of roommate, and the UEs who do not get matched as the SUs. The user benefit is defined as the measure for two agents if they can be matched. Assume UE $m$ $\left(\forall m \in \mathcal{N}\right)$ is matched to UE $k$ $\left(\forall k \in \mathcal{N} \ \text{and} \ m \neq k\right)$, where UE $m$ and UE $k$ respectively act as the RD and the RP in the cooperation, then, they will get the same benefit as
\begin{align}
  U_{m, k} = U_{k, m} = \max_{l_{i, j}} \left( \left( E_m^{\mathbbm{S}} + E_k^{\mathbbm{S}} \right) - \left( E_m^\mathbbm{D} + E_k^{\mathbbm{P}} \right)\right), \label{equ.A_untility}
\end{align}
where $\left( E_m^{\mathbbm{S}} + E_k^{\mathbbm{S}} \right)$ denotes the total energy consumption of UEs $m$ and $k$ when they work independently, whereas $\left( E_m^\mathbbm{D} + E_k^{\mathbbm{P}} \right)$ denotes the total energy consumption when they cooperate. Hence, the proposed benefit function \eqref{equ.A_untility} gives the maximum energy that UEs $m$ and $k$ can save through cooperation. Function (15) is a special case of problem \eqref{equ.c_min_one_function} which can be solved by the proposed SCA algorithm.

By solving function (15), any UE $m$ can rank the other UEs in set $\mathcal{N}$ according to the benefits that they can get from matching with other UEs. The preference list for UE $m$ to choose roommates is set to $\mathcal{PL}_m = \left\{ U_{m, k}^{ \left(1\right)},\cdots,U_{m, N - 1}^{\left(N - 1\right)} \right\}$, where $U_{m, k}^{\left(1\right)}$ and $U_{m, N - 1}^{\left(N - 1\right)}$ denote the most preferred and the least preferred UEs for UE $m$, respectively. Then, the role of each UE can be determined by using the Irving algorithm \cite{irving_1985_algorithm}, which is presented in Appendix A. The RP and RD in each successful matching pair are put into sets $\mathcal{N}^{\mathbbm{P}}$ and $\mathcal{N}^{\mathbbm{D}}$, respectively, but the UEs which can not successfully match a roommate are put into the SU set $\mathcal{N}^{\mathbbm{S}}$, rather than simply rejected as in the classical \textit{Irving algorithm}. As a result, the UE set $\mathcal{N}$ is partitioned into three subsets $\mathcal{N}^{\mathbbm{D}}$, $\mathcal{N}^{\mathbbm{P}}$ and $\mathcal{N}^{\mathbbm{S}}$.

If two UEs $i$ and $j$ ($i \neq j$) both have the incentive to leave their current partners and form a new pair with each other, these two UEs form a blocking pair $\left(i, j\right)$. One can see from Appendix B that the matching result of the proposed user role assignment algorithm is stable as there exists no blocking pair in the matching result.

\subsection{User association}
After obtaining the RD and RP sets, we consider that one RP can serve multiple RDs in a slot. This user association problem can be seen as a multi-to-one two-side matching, which can be transformed into the college admission problem (CAP) \cite{gu_2015_jsac} with the aim to achieve a stable multi-to-one matching.
In the user association problem, the two disjoint UE sets, i.e., the RP set $\mathcal{N}^\mathbbm{P}$ and the RD set $\mathcal{N}^\mathbbm{D}$  correspond to the college set and the student set, respectively. The benefit that RD $i$ can get after associating with RP $j$ is shown as the energy saving for RD $i$ with the help of RP $j$, then, one has
\begin{align}
  U_{i, j}^{\mathbbm{D}} = \max_{l_{i, j}} \left(  E_i^\mathbbm{S} - E_i^\mathbbm{D} \right). \label{equ.B_H_utility}
\end{align}
The benefit or the energy saving which RP $j$ can get is
\begin{align}
U_{j, i}^\mathbbm{P} = \max_{l_{i, j}} \left( \left(E_i^\mathbbm{S} + E_j^\mathbbm{S}\right) - \left( E_i^\mathbbm{D} + E_j^\mathbbm{P} \right) \right).
\label{equ.B_L_utility}
\end{align}
Following \eqref{equ.B_H_utility}, the preference list of RD $i$ is defined as $\mathcal{PL}^\mathbbm{D}_i$, which is a ranking of RD $i$ over all the RPs in set $ \mathcal{N}^{\mathbbm{P}}$ in descending order according to $U_{i, j}^{\mathbbm{D}}$ $\left( \forall j \in \mathcal{N}^{\mathbbm{P}} \right)$. Similarly, following \eqref{equ.B_L_utility}, the preference list of RP $j$ is defined as $\mathcal{PL}^\mathbbm{P}_j$, which is a ranking of RP $j$ over all the RDs in set $\mathcal{N}^{\mathbbm{D}}$ in descending order according to $U_{j, i}^{\mathbbm{P}}$ $\left( \forall i \in \mathcal{N}^{\mathbbm{D}} \right)$.

Up to now, we have mapped the user association problem into the CAP. It can then be solved by using the Gale-Shapley (GS) algorithm, which is presented in Appendix C. Although the obtained matching result is stable, some RDs may not match their preferred RPs \cite{gu_2015_jsac, wang_2019_iot}. Hence we propose the rotation swap operation (RSO) which allows RDs to associate with other RPs with better performance through exchanging their currently matched RPs without losing the stability. Before presenting the RSO, some definitions and notations are given below.

\begin{definition} \label{def.cabal}
  Cabal: a cabal $\mathcal{K}= \left\{k_1, ..., k_x, ..., k_K\right\}$ is a subset of $\mathcal{N}^\mathbbm{D}$, such that $ \Omega\left(k_{x - 1}\right) \succ_{k_x} \Omega(k_x)$ according to $\mathcal{PL}^\mathbbm{D}_{k_x}$, $\forall k_x \in \mathcal{K}$ $\left( x - 1 = K \ when \  x = 1 \right)$.
\end{definition}
\begin{definition} \label{def.accomplice}
  Accomplice: the accomplice set $\mathcal{H}(\mathcal{K})$ of cabal $\mathcal{K}$ is a subset of $\mathcal{N}^\mathbbm{D}$, such that $h \in \mathcal{H}(\mathcal{K})$ if 1) $h \notin \mathcal{K}$, for any $ k_x \in \mathcal{K}$ if $\Omega(k_x) \succ_{h} \Omega(h)$ and $h \succ_{\Omega(k_x)} k_x$, or 2) $h \in \mathcal{K}$ and $h = k_l$ ($\forall k_l \in \mathcal{K}$), for any $k_x \in \mathcal{K}$ and $x \neq l$ if $\Omega(k_x) \succ_{k_l} \Omega(k_{l - 1})$ and $k_l \succ_{\Omega(k_x)} k_{x + 1}$.
\end{definition}

\textbf{Definitions 1} and \textbf{2} show that the RDs in set $\mathcal{H}(\mathcal{K})$ may have prevented the RDs in set $\mathcal{K}$ from matching their more preferred RPs. Next, we propose the falsify operation, which enables the RDs in set $\mathcal{H}(\mathcal{K})$ to help the RDs in set $\mathcal{K}$ match more preferred RPs. The detail is given below.

Let $s \succ_{i} n$ denote that RD $i$ prefers RP $s$ over RP $n$ in $\mathcal{PL}^\mathbbm{D}_i$. Let $\Omega\left(i\right)$ denote the partner of RD $i$ obtained by using the GS algorithm. Assume $\Omega_\text{S}(i)$ denotes the partner of RD $i$ obtained by further performing the RSO. If $\Omega_\text{S}(i) \succeq_{i} \Omega(i)$, we say that matching $\Omega_\text{S}$ is “at least as good as” matching $\Omega$, which is denoted by $\Omega_\text{S} \geq \Omega$. We rewrite the preference list of RD $i$ as $\mathcal{PL}^{\mathbbm{D}}_i = \left( \mathcal{PL}^\mathbbm{D}_L(i), \Omega(i), \mathcal{PL}^\mathbbm{D}_R(i) \right)$, where $\mathcal{PL}^\mathbbm{D}_L(i)$ and $\mathcal{PL}^\mathbbm{D}_R(i)$ denote the RPs that are ranked higher and lower than $\Omega(i)$, respectively. Then, the falsify operation is to move RP $\theta$ $(\forall \theta \in \mathcal{PL}^\mathbbm{D}_L(i))$ from $\mathcal{PL}^\mathbbm{D}_L(i)$ to $\mathcal{PL}^\mathbbm{D}_R(i)$. Let $\pi_r \left( \mathcal{PL}^{\mathbbm{D}}_i \right)$ denote the random permutation of $\mathcal{PL}^{\mathbbm{D}}_i$. Then, we can have the following Lemma.
\begin{lemma}
Let $\mathcal{J} \subseteq \mathcal{N}^\mathbbm{D}$. If all the RDs in set $\mathcal{J}$ ($\forall i \in \mathcal{J}$) submit their falsified lists in the form $\left( \pi_r\left(\mathcal{PL}^\mathbbm{D}_L(i) - \theta\right), \; \Omega\left( i \right), \; \pi_r\left(\mathcal{PL}^\mathbbm{D}_R(i) + \theta\right) \right)$, then $\Omega_S \geq \Omega$.
\end{lemma}
\begin{proof}
  Please refer to Appendix D.
\end{proof}

\textbf{Lemma 1} shows that if the RDs which have matched RPs change the location of their rejected RPs in the preference list, the matching result will not change. The detail of the RSO in presented in \textbf{Algorithm \ref{alg.swap_operation}}, wherein lines 4-10, the RDs in set $\mathcal{H}$ help the RDs in set $\mathcal{K}$ match their more preferred RPs through the falsify operation. In each round of RSO, most of the existing algorithms, e.g. \cite{gu_2015_jsac, wang_2019_iot}, allow only two users to swap resources with each other, while \textbf{Algorithm 1} allows more than two users to swap resources at the same time, thus greatly improving the execution efficiency of the algorithm.

\begin{algorithm}[!h]
  \begin{algorithmic}[1]
    \caption{Rotation Swap Operation (RSO)}
    \label{alg.swap_operation}
    \STATE Let $\Omega = \Omega_0$, where $\Omega_0$ is the matching result obtained by performing the Irving algorithm and the GS algorithm sequentially.
    \STATE Find the largest cabal $\mathcal{K}$ from $\Omega$ by using \textbf{Definition 1}.
    \STATE Find the accomplice $\mathcal{H}(\mathcal{K})$ of cabal $\mathcal{K}$ by using \textbf{Definition \ref{def.accomplice}}.
    \FORALL{RD $i \in \mathcal{K}$}
        \IF{$i \in \mathcal{H}\left(\mathcal{K}\right) - \mathcal{K}$}
            \STATE RD $i$ falsifies its preference list $\mathcal{PL}^\mathbbm{D}_i$ as $\left( \pi_r(\mathcal{PL}^\mathbbm{D}_L(i) - \theta), \; \Omega(i), \; \pi_r(\mathcal{PL}^\mathbbm{D}_R(i) + \theta) \right)$, where $\theta = \left\{ c \; | \; \Omega(m) \in \Omega(\mathcal{K}), \; m \succ_{c} m + 1 \right\}$.
        \ELSE
            \STATE RD $i$ ($i = i_l, \forall i_l \in \mathcal{K}$) falsifies its preference list $\mathcal{PL}^\mathbbm{D}_i$ as $\left( \pi_r(\mathcal{PL}^\mathbbm{D}_L(i) - \theta), \; \Omega(i - 1), \; \pi_r(\mathcal{PL}^\mathbbm{D}_R(i) + \theta)\right)$, where $\theta = \left\{ c \; | \; \Omega(m) \in \Omega(\mathcal{K}), \; w \succ_{l} \Omega(m - 1)\right.$, $\left. l \succ_{w} m + 1 \right\}$.
        \ENDIF
    \ENDFOR
    \RETURN The modified preference lists.
  \end{algorithmic}
\end{algorithm}

\subsection{Resource allocation} \label{sec.optimal_offloading}
Now, the UEs in the system have been divided into multiple \textit{cooperation groups}, each consisting of one RP and multiple RDs. Therefore, problem \eqref{equ.min_function} can be rewritten as
\begin{alignat}{3}
  &\min_{l_{i, j}, \; p_{i, j}} \; && E^\mathbbm{P}_j + \sum_{i \in \mathcal{M}_j} E_i^\mathbbm{D} \label{equ.c_min_one_function} \\
  &\mbox{s.t.} && \eqref{equ.sr_offloading_size_cons}, \; \eqref{equ.transmit_power}. \notag
\end{alignat}

We note that problem \eqref{equ.c_min_one_function} is non-convex due to the non-convexity of constraint \eqref{equ.transmit_power}.
To address this issue, we define $N_0 = \sigma^2 / g_{i, j}$ and $\alpha = 2^{1/{\tau w}}$,
and convert constraint \eqref{equ.transmit_power} into the following form
\begin{equation}
  p_{i, j} = N_0 \left( \alpha^{\sum_{i \in \mathcal{M}_j} l_{i, j}} - \alpha^{ \sum_{n \in \mathcal{M}_j\backslash\left\{i\right\}} l_{n, j} } \right) \leq P_i^\text{Max}.
\end{equation}
We obtain the partial derivatives of the $p_{i, j}$ with respect to $l_{i, j}$ and $l_{n, j}$ as
\begin{equation}
    \triangledown_{l_{i, j}} p_{i, j} = N_0 \ln \alpha \left( \alpha^{\sum_{i \in \mathcal{M}_j} l_{i, j}} \right).
    \label{equ.l_i_derivate}
\end{equation}
and
\begin{equation}
    \triangledown_{l_{n, j}} p_{i, j} = N_0 \ln \alpha \left( \alpha^{\sum_{i \in \mathcal{M}_j} l_{i, j}} - \alpha^{\sum_{n \in \mathcal{M}_j\backslash\left\{i\right\}} l_{n, j}} \right).
    \label{equ.l_n_derivate}
\end{equation}
respectively. Then, we get the following Lemma.
\begin{lemma}
$\triangledown_{l_{i, j}} p_{i, j}$ and $\triangledown_{l_{n, j}} p_{i, j}$ are Lipschitz continuous on $l_{i, j}$ and $l_{n, j}$ with the constants $L_C$ and $L_F$, respectively, where $L_C$ and $L_F$ are the Lipschitz constants.
\end{lemma}
\begin{proof}
  Please refer to Appendix E.
\end{proof}

If the first derivative of a non-convex function is Lipschitz continuous on the variables with the corresponding Lipschitz constants and each variable is nonempty, closed, and convex, then, the non-convex function can be converted to a convex one by SCA \cite{wu_2018_GC}. \textbf{Lemma 2} indicates that $ p_{i, j} $ meets this rule and thus can be converted to a convex one by using SCA.

Let $\mathcal{D}_j \triangleq \left( l_{1, j}, ..., l_{\left| \mathcal{M}_j \right|, j} \right)^T$ denote the offloading strategy profile of the RDs in set $\mathcal{M}_j$. Let $\delta = 1, 2, ... $ denote the iterative numbers. Then, $l_{i, j}[\delta]$ represents the amount of task data offloaded from RD $i$ to RP $j$ in the $\delta^{\textbf{th}}$ iteration. By using eqs. \eqref{equ.l_i_derivate} and \eqref{equ.l_n_derivate}, we derive the gradient of $ p_{i, j} $ on $\mathcal{D}_j [\delta]$ as
\begin{equation}
  \triangledown_{\mathcal{D}_j [\delta]} p_{i, j} = \left( \frac{\partial p_{i, j}}{\partial l_{1, j}[\delta]}, \frac{\partial p_{i, j}}{\partial l_{2, j}[\delta]}, ..., \frac{\partial p_{i, j}}{\partial l_{\left| \mathcal{M}_j \right|, j}[\delta]} \right)^T.
\end{equation}
According to \cite{Facchinei_2014_icassp}, a strongly convex function to approximate $p_{i, j}$ can then be constructed as
\begin{align}
  \widetilde{p}_{i, j} \triangleq & \; p_{i, j} [\delta] + \left( \mathcal{D}_j - \mathcal{D}_j[\delta] \right)^T \triangledown_{\mathcal{D}_j [\delta]} p_{i, j} + \frac{\lambda}{2} \left|| \mathcal{D}_j - \mathcal{D}_j[\delta] \right||_2 + \notag  \\
  & \frac{1}{2} \left( \mathcal{D}_j - \mathcal{D}_j[\delta] \right)^T \left( \mathcal{D}_j - \mathcal{D}_j[\delta] \right) \triangledown_{\mathcal{D}_j [\delta]}^2 p_{i, j}.
\end{align}
where $\lambda \geq 0$. Now, we can approximate constraint \eqref{equ.transmit_power} as
\begin{equation}
  \widetilde{p}_{i, j} \leq P_i^\text{Max}, \;\; \forall i \in \mathcal{M}_j. \label{equ.transfromed}
\end{equation}

By replacing constraint \eqref{equ.transmit_power} with constraint \eqref{equ.transfromed}, problem \eqref{equ.c_min_one_function} is converted to a convex one and can be solved by using convex optimization tools, e.g., CVX \cite{wu_2018_GC}.
The detail of the SCA algorithm is presented in Appendix F.

\subsection{Overall algorithm}

\begin{algorithm}[!h]
  \begin{algorithmic}[1]
    \caption{Overall algorithm to solve problem (14)}
    \label{alg.overall_alg}
    \STATE Perform the \textit{user role control} to limit the UEs with available energy less than the threshold, i.e.,  $E_m^{\text{Ax}} \leq E_m^{\text{min}}$ to become RPs.
    \STATE Perform the \textit{Irving algorithm} to partition the UE set $\mathcal{N}$ into three subsets $\mathcal{N}^{\mathbbm{P}}$, $\mathcal{N}^{\mathbbm{D}}$, $\mathcal{N}^{\mathbbm{S}}$.
    \STATE  Perform the \textit{GS algorithm} to establish the association between the UEs in set $\mathcal{N}^{\mathbbm{P}}$ and the UEs in set $\mathcal{N}^{\mathbbm{D}}$.
    \REPEAT
        \STATE Perform \textit{Algorithm 1}, i.e., the RSO to find \textit{cabal} and falsify the \textit{preference lists} of the UEs in the \textit{cabal}.
        \STATE Perform the \textit{GS algorithm} by using the modified \textit{preference lists} to update the UE association, that is
        \IF{there exist no $j \in \mathcal{N}^\mathbbm{P}$ such that $\Omega(j) \succ_{j} \Omega_\text{S}(j)$}
            \STATE update the UE association.
        \ENDIF
    \UNTIL{the matching result has no cabal.}
    \STATE  Perform the \textit{SCA} algorithm for each \textit{cooperation group} to obtain the resource allocation.
  \end{algorithmic}
\end{algorithm}

We present the overall algorithm for implementing the MUCC in \textbf{Algorithm \ref{alg.overall_alg}}. The algorithm first performs the user role control operation (line 1), which limits the UEs with the remaining energy less than the threshold to become RPs, thus satisfying constraint \eqref{fig.total_cons}. The complexity of the SCA algorithm (line 11) is related to the number of iterations $\delta$ as well as the complexity of the solver for the convex problems \cite{wu_2018_GC}. Since the adopted CVX toolbox is based on the \textit{standard interior point method}, whose complexity is \textbf{O}$\left( a_j^3 \right)$ \cite{6409501}, the complexity of the SCA algorithm is thus \textbf{O}$\left( \delta a_j^3 \right)$. As analyzed in Appendix G, the computational complexity of \textbf{Algorithm 2} is \textbf{O}$\left(N^2 + \vert \mathcal{N}^{\mathbbm{P}} \vert \cdot \vert \mathcal{N}^{\mathbbm{D}} \vert +  \delta a_j^3\right)$.

The proposed MUCC scheme is enabled by D2D communications. The effective range of the D2D communications is limited to one-hop, and a cellular base station (BS) is required to establish and manage the D2D links. Therefore, \textbf{Algorithm \ref{alg.overall_alg}} can be performed at a BS in a \textit{centralized} manner, as the BS can obtain the channel state information (CSI), and the information of available energy and computing resource of all UEs through dedicated feedback channels. However, inaccurate CSI estimation will prevent the system from achieving the optimal performance. Please refer to Appendix G for a detailed analysis.

\section{Simulation Results}
We consider a $100 \times 100 \; \text{m}^2$ rectangular area, where multiple UEs are randomly distributed. This is a general WLAN case that enables distributed UEs to perform D2D communications and further perform cooperative computing. Other parameters used in the simulation are given below. The number of CPU cycles required to execute one task bit is $C_m = 500$, and the effective capacitance coefficient is $\gamma_m = 10^{-28}$ \cite{lin_2019_coml}. The noise power is $\sigma^2 = 10^{-9}$. The maximum transmit power of the UE is set to 0.1 W. The duration of a time slot is $\tau = 0.2$ s. The bandwidth is $w = 1$ MHz. The channel power gain is set to $g_{i, j} = \zeta_{i, j} / d^{\beta}$ \cite{gu_2015_jsac}, where $d$ is the distance between two UEs (in meters), $\beta = 3$ is the power-scaling path-loss factor, and $\zeta_{i, j}$ follows an exponential distribution with a unit mean, which captures the fading and shadowing effects. The maximum connection constraint of the UE is $a_m = 2$.

To verify the effectiveness of the proposed MUCC scheme, we compare our proposed algorithm with the exhaustive search (ES) algorithm. Although the ES algorithm can find the optimal solution of problem \eqref{equ.min_function}, the computational complexity is too high to be used in practice. The number of UEs in the system is set to less than 12. The simulation results in shown in Fig. \ref{fig.energy_consumption}. One can see from Fig. \ref{fig.energy_consumption} that the performance loss of our proposed algorithm is small compared to the optimal solution obtained by using the ES algorithm, but our algorithm far outperforms the local computing method. It indicates that the proposed algorithm achieves a better compromise between computational complexity and optimal performance.

\begin{figure}[!t]
    \centering
    \includegraphics[width=0.75\linewidth]{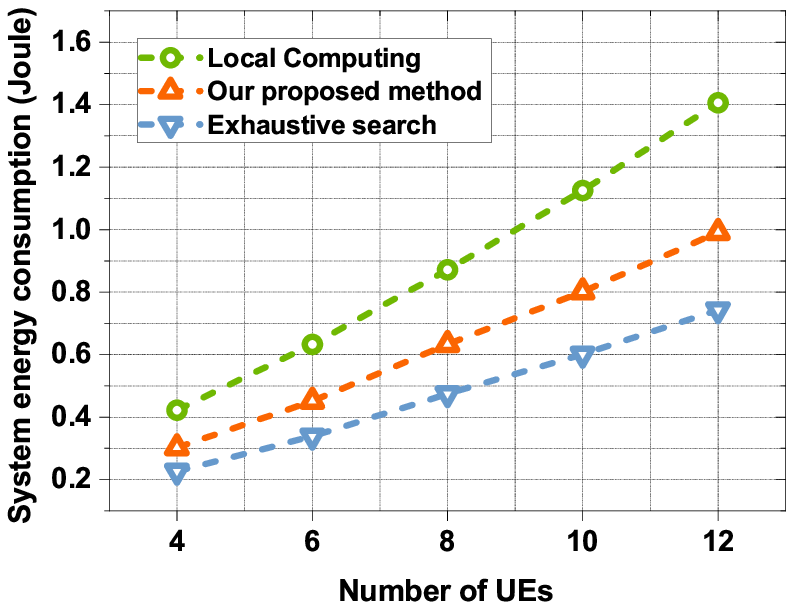}
    \caption{The overall energy versus the number of UEs.}
    \label{fig.energy_consumption}
\end{figure}

\begin{figure}[!t]
    \centering
    \includegraphics[width=0.75\linewidth]{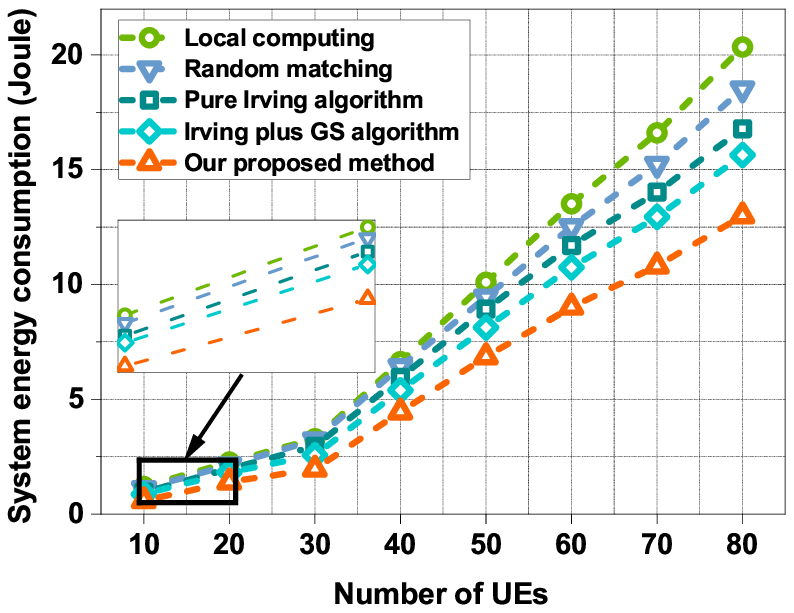}
    \caption{The overall energy versus the number of UEs.}
    \label{fig.energy_all}
\end{figure}

Next, we compare our proposed algorithm with the one-to-one random-matching algorithm, the pure Irving algorithm, and the Irving plus GS algorithm in a network with a larger number of UEs. The amount of task data for each UE is randomly distributed in $[0, 1]$ Mbits. Fig. \ref{fig.energy_all} shows the overall energy versus the number of the UEs. One can see that our proposed algorithm achieves better energy efficiency over the other algorithms in all cases. The reason behind is that the proposed algorithm realizes the global user matching and resource allocation according to the available computing resources to the RPs, the required computing resources for the RDs, and the channel conditions between them, thus greatly reducing the overall energy consumption of the system.

\begin{figure}[!t]
  \subfigure[Before performing the RSO.] {
    \includegraphics[width=0.466\linewidth]{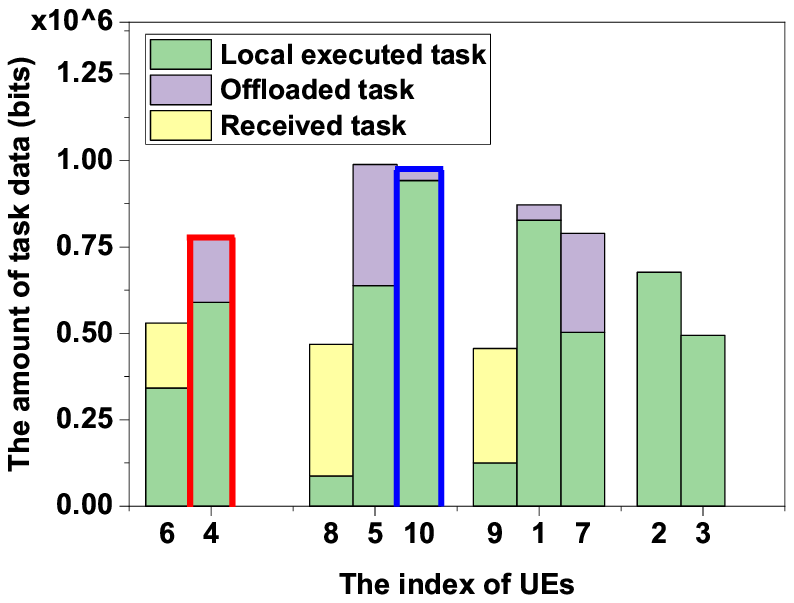}
    \label{fig.deployment_energy}
  }
  \subfigure[After performing the RSO.] {
    \includegraphics[width=0.466\linewidth]{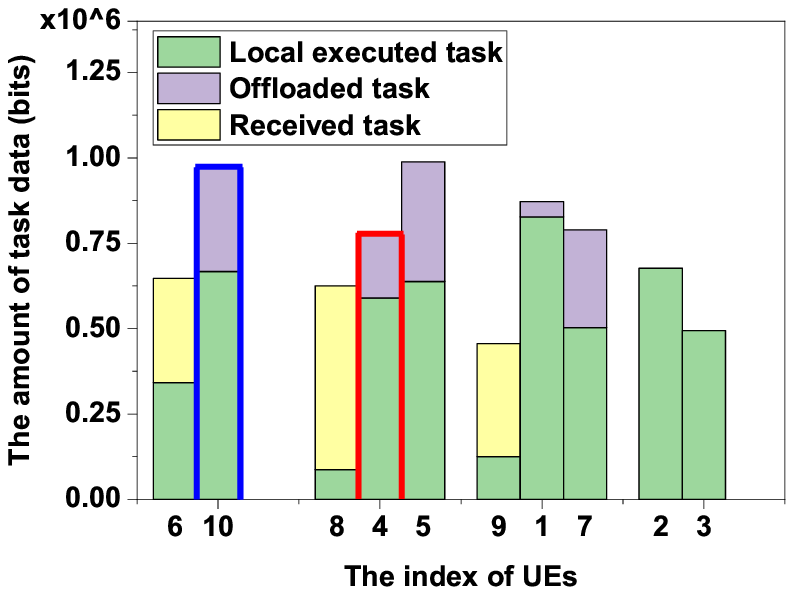}
    \label{fig.after_swap_deployment_energy}
  }
  \caption{Task offloading results.}
  \label{fig.after_all_tasks}
\end{figure}

Finally, we show how the proposed RSO influences the system performance when there are 10 UEs. Figs. \ref{fig.after_all_tasks}\subref{fig.deployment_energy} and \ref{fig.after_all_tasks}\subref{fig.after_swap_deployment_energy} show the amount of offloaded data before and after performing the RSO, respectively.
One can see that after performing the RSO, the amount of tasks offloaded by both UEs 4 and 10 are increased, thus reducing their energy consumption and further reducing the total energy consumption of the system. The reason behind is that the RDs in set $\mathcal{H}$ help the RDs in set $\mathcal{K}$ match their more preferred RPs through the falsify operation (lines 4-10 of \textbf{Algorithm 1}). Thus, the RSO can further decrease the overall energy consumption of the system through benefiting the RDs.

\section{Conclusion and Future Works}
This paper presents a MUCC scheme based on matching theory that answers the questions about who and how to cooperate with. In the future work, we will extend the single-hop MUCC scheme to the multi-hop MUCC scheme, where the task data of a RD could be forwarded by multiple relay UEs to reach the RP. In such a case, the energy consumption of all the relay UEs along a route should be considered when balancing the energy consumption for task execution and that for data transmission. In addition, meeting the delay requirements for user computing tasks is another huge challenge.

\begin{appendices}
\section{The User Role Assignment Algorithm}
\begin{algorithm}[H]
  \begin{algorithmic}[1]
    \caption{User Role Assignment}
    \label{alg.one2one}
    \REQUIRE let $\mathcal{UL}=\mathcal{N}$, where $\mathcal{UL}$ is the set of the unmatched UEs in a multi-user cooperative computing system.
    \WHILE{$\mathcal{UL}$ is not empty}
      \STATE Choose UE $m$ $(\forall m\in\mathcal{UL})$ from set $\mathcal{UL}$;
      \IF{UE $m$ is rejected by all the UEs in $\mathcal{PL}_m$}
        \STATE Remove UE $m$ from set $\mathcal{UL}$;
        \STATE Go back to \textit{Step 1};
      \ENDIF
      \STATE UE $m$ proposes to match UE $k$ $(\forall k\in\mathcal{PL}_m)$, which is the most favorite partner of UE $m$ that does not reject UE $m$;
      \IF{UE $k$ holds no proposal \OR prefers UE $m$ to previous proposal UE $s$ $(\forall s\in\mathcal{PL}_k)$ }
        \STATE UE $k$ matches UE $m$ and rejects UE $s$;
        \STATE Remove UE $m$ from set $\mathcal{UL}$;
        \STATE Add UE $s$ in set $\mathcal{UL}$;
      \ELSE
        \STATE UE $k$ rejects UE $m$;
      \ENDIF
    \ENDWHILE
    \STATE Each successfully matched UE retains the matching UE and rejects all the other UEs in the preference list;
    \FORALL{successful matching pairs in the system}
        \STATE Put the UEs acting as the resource providers (RPs) in the matching pairs into set $\mathcal{N}^{\mathbbm{D}}$;
        \STATE Put the UEs acting as the resource demanders (RDs) in the matching pairs into set $\mathcal{N}^{\mathbbm{P}}$;
    \ENDFOR
    \STATE Put the non-matched UEs into set $\mathcal{N}^{\mathbbm{S}}$;
    \RETURN UE sets $\mathcal{N}^{\mathbbm{D}}$, $\mathcal{N}^{\mathbbm{P}}$, and $\mathcal{N}^{\mathbbm{S}}$.
  \end{algorithmic}
\end{algorithm}

\section{The Proof of Stability of URA Algorithm}
We prove this Lemma by contradiction. Assuming that the matching result is not a stable matching, then at least one blocking pair $(m, k)$, $\forall m,k \in \mathcal{N}$ and $m \neq k$ exists in the matching result. We analyze the following two cases.
\begin{enumerate}
    \item In the first case, we consider that UE $k$ is unmatched and never receives a request from UE $m$. It indicates that UE $m$ prefers its current partner to UE $k$. Therefore, pair ($m, k$) cannot constitute the \textit{blocking condition} (given in the paragraph above Lemma 1 of the paper).
    \item In the second case, we consider that UE $k$ is matched successfully and the unilaterally \textit{blocking condition} that UE $k$ prefers UE $m$ to its current partner is satisfied. In such a case, UE $m$ is either unmatched or prefers UE $k$ to the current partner. Then, UE $m$ should request computing resource from UE $k$ or pair ($m, k$) should be removed. In either case, all the UEs located after UE $m$ in the preference list of UE $k$ (including the current partner of UE $k$) will be rejected by UE $k$, which causes a contradiction.
\end{enumerate}

\section{The Gale-Shapley Algorithm}
  \begin{algorithm}[H]
    \begin{algorithmic}[1]
      \caption{The Gale-Shapley Algorithm}
      \label{alg.many2one}
      \REQUIRE let $\mathcal{UL} = \mathcal{N}^\mathbbm{D}$, where $\mathcal{UL}$ is the set of the unmatched UEs.
      \WHILE{$\mathcal{UL}$ is not empty}
        \STATE Choose RD $i$ ($\forall i \in \mathcal{UL}$) from set $\mathcal{UL}$;
        \IF{RD $i$ is rejected by all RPs in $\mathcal{PL}_i^\mathbbm{D}$}
          \STATE Remove RD $i$ from set $\mathcal{UL}$;
          \STATE Go back to \textit{Step 1};
        \ENDIF
        \STATE RD $i$ proposes to match RP $j$ ($\forall j \in \mathcal{PL}_i^\mathbbm{D}$), which is the most favorite RP of RD $i$ that does not reject RD $i$;
        \IF{RP $j$ holds the number of proposals less than the maximum limit $a_j$}
          \STATE RP $j$ matches RD $i$;
          \STATE Remove RD $i$ from set $\mathcal{UL}$;
        \ELSIF{RP $j$ holds more than $a_j$ proposals \AND prefers RD $i$ to one of its current proposal RDs}
          \STATE RP $j$ selects the least preferred RD $m$ ($\forall m \in \mathcal{PL}^{\mathbbm{P}}_j$) from the current proposal SDs and rejects it;
          \STATE Add RD $m$ to set $\mathcal{UL}$;
          \STATE RP $j$ matches RD $i$;
          \STATE Remove RD $i$ from set $\mathcal{UL}$;
        \ELSE
          \STATE RP $j$ rejects RD $i$.
        \ENDIF
      \ENDWHILE
      \RETURN UE association.
    \end{algorithmic}
  \end{algorithm}

\section{The Proof of Lemma 1}
Considering that all the RPs in $\mathcal{PL}^\mathbbm{D}_L(i)$ have rejected the request of RD $i$, we prove this Lemma by analyzing the following two cases.
\begin{enumerate}
    \item In the first case, RD $i$ does not propose to match the RPs in $\mathcal{PL}^\mathbbm{D}_R(i)$. Therefore, the random permutation of $\mathcal{PL}^\mathbbm{D}_L(i)$ and $\mathcal{PL}^\mathbbm{D}_R(i)$ has no effect on the result.
    \item In the second case, we assume that the \textit{shift operation} is performed for the RPs in $\mathcal{PL}^{\mathbbm{D}}_L (i)$. Hence, part of them is moved from $\mathcal{PL}^\mathbbm{D}_L (i)$ to $\mathcal{PL}^\mathbbm{D}_R (i)$, and, the newly obtained $\mathcal{PL}^\mathbbm{D}_L(i)$ and $\mathcal{PL}^\mathbbm{D}_R(i)$ are randomly permuted. As RD $i$ still holds the same partner as before, the current matching $\Omega_{\text{S}}$ is at least the same as $\Omega$.
\end{enumerate}

\section{The Proof of Lemma 2}
Let $x=l_{i, j}$ and $G\left( l_{i, j} \right) = \triangledown_{l_{i, j}} F_{i, j}$. Due to the monotonic increasing of $G\left( l_{i, j} \right)$, one can get
\begin{equation}
\lim_{x \to 0} G\left( x \right) \leq G\left( x \right) \leq \lim_{x \to L_{i}} G\left( x \right).
\end{equation}
For any $x_1, x_2 \in \left[ 0, L_{i} \right]$, one can further derive that
\begin{equation}
\left| G\left( x_1 \right) - G\left( x_2 \right) \right| \leq G\left( L_{i} \right) - G\left( 0 \right).
\end{equation}
Since $\left| x_1 - x_2 \right| \leq L_{i}$, it is easy to derive that
\begin{equation}
\left| G\left( x_1 \right) - G\left( x_2 \right) \right| \leq L_C L_{i},
\end{equation}
where
\begin{equation}
L_C = \frac{ G\left( L_{i} \right) - G\left( 0 \right)}{L_{i}},
\end{equation}
is the Lipschitz constant of $G\left( x \right)$.

Now, we have proved that $\triangledown_{l_{i, j}} F_{i, j}$ is Lipschitz continuous on $l_{i, j}$ with the Lipschitz constant $L_C$. The Lipschitz continuity of $\triangledown_{l_{n, j}} F_{i, j}$ can be proved in a similar way which is omitted here.

\section{The SCA Algorithm}

\begin{algorithm}[!h]
\begin{algorithmic}[1]
\caption{The SCA based algorithm to solve problem (18)}
\STATE Initialize $\delta=0$.
\STATE Find a feasible $D_j [\delta]$.
\REPEAT
\STATE For a given $D_j [\delta]$, find the optimal $D_j^* [\delta]$ by solving problem (18) with constraint (24);
\STATE $D_j [\delta + 1] = D_j^* [\delta] + \Delta[\delta]$;
\STATE $\delta = \delta + 1$;
\UNTIL{the objective function of problem (18) converges, or the maximum number of iterations is reached.}
\RETURN $D_j [\delta]$ as well as the value of the objective function of problem (18).
\end{algorithmic}
\end{algorithm}

\section{The Computational Complexity of Overall Algorithm}
The overall algorithm involves the following four algorithms.
\begin{enumerate}
    \item The User Role Assignment (URA) algorithm is based on the Irving algorithm with the computational complexity of \textbf{O}($N^2$) [7], where $N$ denotes the number of the involved user devices.

    \item The Gale-Shapley (GS) algorithm is sure to converge after at most $\vert \mathcal{N}^{\mathbbm{P}} \vert \cdot \vert\mathcal{N}^{\mathbbm{D}} \vert$ iterations [8], where $\mathcal{N}^{\mathbbm{P}}$ and $\mathcal{N}^{\mathbbm{D}}$ denote the RP set and the RD set, respectively.

    \item The Rotation Swap Operation (RSO) uses a depth-first traversal to find cabals. Since each user is accessed at most once during a lookup, the complexity of the RSO is \textbf{O}($N$) [9].

    \item The Successive Convex Approximation (SCA) algorithm is used to find the optimal task assignment strategy within a user group. It is noted that a user group consists of one RP $j$ and $a_j$ RDs, and $\delta$ denotes the iterator number.
    \begin{enumerate}
        \item As stated in ref. \cite{wu_2018_GC} and \cite{6409501}, the method of SCA algorithms is to iteratively solve a series of convex problems to eventually approximate the optimal solution of the original problem. Therefore, the complexity of the SCA algorithm is related firstly to the number of iterations $\delta$ and secondly to the complexity of solving the convex problems.

        \item In this paper, we use the CVX toolbox as in \cite{wu_2018_GC} and \cite{6409501} to solve the convex problems in each iteration. From ref. \cite{6409501} (pp. 697), we know that the CVX toolbox is implemented by using the \textit{standard interior point method}, whose complexity is \textit{cubic} with respect to the dimensionality of the input space
        In the worst case, the input size can be considered as the total number of users $N$. Thus the complexity of the CVX toolbox to solve a convex problem in each iteration is \textbf{O}$\left( N^3 \right)$.

        \item Based on the above analysis, we can give the complexity of the whole SCA algorithm as \textbf{O}$\left( \delta N^3 \right)$.
    \end{enumerate}
\end{enumerate}

According the above analysis, the computational complexity of the overall algorithm is given by \textbf{O}$\left(N^2 + \vert \mathcal{N}^{\mathbbm{P}} \vert \cdot \vert\mathcal{N}^{\mathbbm{D}} \vert +  \delta N^3 \right) = \textbf{O}\left( \delta N^3  \right)$.

\end{appendices}

\bibliographystyle{IEEEtran}
\bibliography{IEEEabrv}

\end{document}